\documentclass[a4paper,11pt]{article}

\usepackage{amsfonts,amssymb,amsmath,amsthm,dsfont,xfrac,xspace}
\usepackage{fullpage}
\usepackage{graphicx}
\usepackage{subfig}
\usepackage{float}
\usepackage{cite}
\usepackage{hyperref}
\graphicspath{{./},{fig/}}
\makeatletter
\let\NAT@parse\undefined
\makeatother
\usepackage[sort&compress, numbers]{natbib}
\usepackage{hyperref}

\usepackage{algpseudocode}

\newcommand{\AddStatexIndent}{%
	\hskip\algorithmicindent}

\makeatletter
\algnewcommand{\LineComment}[1]{\Statex \hskip\ALG@thistlm // #1}
\algnewcommand{\LineCommentIndent}[1]{\Statex \hskip\ALG@thistlm\AddStatexIndent // #1}
\makeatother

\newsavebox{\ieeealgbox}
\newenvironment{boxedalgorithmic}
  {\begin{lrbox}{\ieeealgbox}
   \begin{minipage}{\dimexpr\columnwidth-2\fboxsep-2\fboxrule}
   \begin{algorithmic}[1]}
  {\end{algorithmic}
   \end{minipage}
   \end{lrbox}\noindent\fbox{\usebox{\ieeealgbox}}}

\usepackage{listings}
\usepackage{xcolor}
\definecolor{mGreen}{rgb}{0,0.6,0}
\definecolor{mGray}{rgb}{0.5,0.5,0.5}
\definecolor{mPurple}{rgb}{0.58,0,0.82}
\definecolor{backgroundColour}{rgb}{0.95,0.95,0.92}

\lstdefinestyle{CStyle}{
	backgroundcolor=\color{backgroundColour},   
	commentstyle=\color{mGreen},
	keywordstyle=\color{magenta},
	numberstyle=\tiny\color{mGray},
	stringstyle=\color{mPurple},
	basicstyle=\footnotesize,
	breakatwhitespace=false,         
	breaklines=true,                 
	captionpos=b,                    
	keepspaces=true,                 
	numbers=left,                    
	numbersep=5pt,                  
	showspaces=false,                
	showstringspaces=false,
	showtabs=false,                  
	tabsize=2,
	language=C
}

\newcommand{\bsy}[1]{\boldsymbol{#1}}

\newtheorem{amsprop}{Proposition}
\newtheorem{amsdef}{Definition}

\newcounter{algoCounter}

\renewcommand\thesection{\Roman{section}}
\renewcommand\thesubsection{\thesection.\Alph{subsection}}
\renewcommand\thesubsubsection{\thesubsection.\arabic{subsubsection}}

\title{Monitoring Large Crowds With WiFi: A Privacy-Preserving Approach} 

\author{Jean-Fran\c cois~Determe$^*$,
	Sophia~Azzagnuni$^*$,
	Utkarsh~Singh$^*$,
	Fran\c cois~Horlin$^*$,\\
	and~Philippe~De~Doncker% <-this % stops a space
	\thanks{$^*$All authors are with the OPERA Wireless Communications Group, Université libre de Bruxelles, 1050 Brussels, Belgium. Corresponding e-mail: Jean-Francois.Determe@ulb.be. Innoviris funded Jean-François Determe and Utkarsh Singh. }% <-this % stops a space
}

\begin{document}
\maketitle

\begin{minipage}{0.9 \linewidth}
	\textbf{\underline{IEEE copyright notice}} --- Paper accepted in \textit{IEEE Systems}, current DOI: 10.1109/JSYST.2021.3139756 --- \copyright \, 2022 IEEE. Personal use of this material is permitted.  Permission from IEEE must be obtained for all other uses, in any current or future media, including reprinting/republishing this material for advertising or promotional purposes, creating new collective works, for resale or redistribution to servers or lists, or reuse of any copyrighted component of this work in other works.
\end{minipage}

\begin{abstract}
This paper presents a crowd monitoring system based on the passive detection of probe requests. The system meets strict privacy requirements and is suited to monitoring events or buildings with a least a few hundreds of attendees. We present our counting process and an associated mathematical model. From this model, we derive a concentration inequality that highlights the accuracy of our crowd count estimator. Then, we describe our system. We present and discuss our sensor hardware, our computing system architecture, and an efficient implementation of our counting algorithm---as well as its space and time complexity. We also show how our system ensures the privacy of people in the monitored area. Finally, we validate our system using nine weeks of data from a public library endowed with a camera-based counting system, which generates counts against which we compare those of our counting system. This comparison empirically quantifies the accuracy of our counting system, thereby showing it to be suitable for monitoring public areas. Similarly, the concentration inequality provides a theoretical validation of the system.  
\end{abstract}

\section{Introduction} \label{sec:intro}
Crowd counting systems count crowd numbers in specific geographical areas and provide these numbers to personnel responsible for their analysis. What follows reviews some use cases of crowd counting systems.\\

In the particular case of public events, event managers have expressed their interest in leveraging modern counting technologies to i) monitor events in real time \cite[Sec.~7]{martella2017current}, ii) predict crowd counts in the future \cite[Sec.~5.1.1]{martella2017current}, and iii) perform post-analyses, to analyze the causes of overcrowding after its occurrence. In particular, computing real-time crowd densities in strategic areas allows security managers to decide whether an event has reached its maximum capacity \cite{still2014introduction, martella2017current}. Crowd count time series can be fed into forecasting algorithms to predict overcrowding \cite{determe2020forecasting, singh2020crowd}---which allows security personnel to execute countermeasures anticipatedly.\\

Crowd management in large events is not the only endeavor that benefits from crowd counting systems. For example, we installed the crowd counting system this paper presents on one of the main commercial streets of Brussels, namely \textit{Rue Neuve} (\textit{Nieuwestraat} in Dutch), to estimate attendance during winter sales. It has been reinstalled in the same street to track attendance as Covid-19 lockdown measures get incrementally relaxed. Finally, we also installed our monitoring system in the largest library of our university: the Humanities library.\\

To summarize, the use cases of crowd counting systems include the monitoring of \textit{i)} public events (to prevent overcrowding) \textit{ii)} commercial streets (to estimate attendance) \textit{iii)} public places wherein some degree of social distancing should be attained and \textit{iv)} public buildings (e.g., university libraries).

\subsection{Related work} \label{subsec:stateoftheart} 

\subsubsection{Mainstream approaches to crowd counting}
This section reviews the main approaches to crowd counting. Because the measurement principles underlying some of these approaches make their field of applicability different from that of the system of this manuscript, no extensive details about them are provided. The main recent works contending with this manuscript are more thoroughly commented in the next section. The reviewed approaches below are mainly inspired from \cite[Sec.~3]{cecaj2021sensing} and \cite[Sec.~1.1]{kaminska2019indoor}. Another excellent review of recent works in crowd monitoring making use of WiFi is \cite[Sec.~2 and Table~1]{uras2020pma}. Other more general reviews are \cite{singh2020crowdreview}, \cite[Sec.~1.1]{kaminska2019indoor} and \cite[Sec.~2 and Table~1]{ryu2020wifi}.\\

A common counting approach is cameras, traditional or thermal \cite{gade2014thermal}. Cameras typically suffer from privacy concerns; from a technical point of view, they suffer from line-of-sight obstructions, non-ideal meteorological conditions, low illumination and high contrast. Thermal cameras are less sensitive to all these issues except for line-of-sight obstructions.\\

Sensor networks are another option. These represent a vast body of approaches. For example, CO2 sensors are an option but are sensitive to air renewal. Acoustic sensors are another option and can be combined with the former one \cite{agarwal2014algorithms}. Another approach, which shall be more extensively developed in the following subsection, is a network of sensors measuring their pairwise communication channels and computing signal attenuation to infer crowd density.\\

Aggregated mobile phone data, which provide time series of numbers of people per geographical cell \cite{calabrese2014urban} are another interesting avenue of information for estimating crowd counts. However, the granularity of these data is sometimes too coarse, making them unsuitable to estimate the attendance of, e.g., a university library. \\

A modern and newer solution is based on WiFi monitoring systems. Such systems wait for individuals' smartphones to connect to a network or install an application (cooperative approach), or they monitor over-the-air beacon signals sent by these smartphones (non-cooperative approach). This solution is newer than most of the previous ones because, two decades ago, no one had Wi-Fi or Bluetooth-enabled electronic devices. The subsection that follows discusses this solution extensively.\\

Another bleeding-edge approach is the monitoring of the electromagnetic spectrum \cite{donelli2021crowd}. This solution is non-cooperative and consists in monitoring frequency bands used by telco operators and their customers to make calls, send text messages and have mobile internet access. We do not have the legal expertise to determine to what extent licensed frequency bands can be monitored in each European country, however. \\

Finally, another emerging technique is the use of modern radars to count people or estimate their flow \cite{yildirim2021super}. These radars are non-cooperative systems and can even reuse existing over-the-air transmissions for radar processing (they are then called passive radars). The feasibility of this last solution for dense crowds remains an open topic of research, however.

\subsubsection{The most relevant former works on crowd counting}
Several works from other teams have tackled the problem of crowd counting and share similarities with the present manuscript. When possible, this section presents accuracy figures for surveyed papers. Table~\ref{tab:comparisonWorks} summarizes the main features of the counting systems that are the main contenders to that presented in this paper. Section~\ref{sec:comparisonContenders} compares them with the system this paper proposes.\\

The authors of \cite{denis2020large, kaya2020large, denis2020sensing} deployed tens of nodes across rooms to be monitored and make them communicate with one another. The received powers for all communication links are a proxy for the number of attendees, because human bodies attenuate WiFi signals (the higher the attenuation, the higher the number of people). This solution is fully non-cooperative, is compatible with low numbers of attendees (<100 people), is not affected by MAC address randomization and can be calibrated easily when the monitored room is empty. However, nodes must be at a low height (< 2 meters) for human bodies to attenuate signals. Moreover, tens of nodes are necessary to monitor a single room (they installed approximately one node per 15 to 40 square meters based on \cite[Fig.~2,~13, and~24]{denis2020large}). Besides, their counting errors are higher than ours:  Results in \cite[Fig.~11]{denis2020large} indicate a mean relative error ranging from 14.6 \% to 22.1 \% depending on the training method.\\

The work \cite{kaminska2019indoor} proposes a crowd monitoring solution for user localization in large buildings. They rely on clients connected to access points they control. Therefore, their approach is partially cooperative. As a result, they depend on users willingly connecting to their access points but do not have to deal with MAC randomization issues. Their method estimates the positions of individuals in a x-y plane for each floor and crowd counting is a byproduct. While the authors have arguments to claim that their method should not be sensitive to high crowd densities \cite[Sec.~5.1]{kaminska2019indoor}, their experiments cover environments hosting less than 100 people. Their accuracy figures range from 90 to 96 \% depending on the area monitored. A similar work is \cite{zhang20193d}.\\

An older and seminal work is \cite{musa2012tracking} in which the authors emulate APs for common service set identifiers (SSIDs) and SSIDs present in the information elements (IEs) of detected PRs. They also send request to send (RTS) packet injection. \cite{musa2012tracking} thus describes an active scanning system.\\

Another work is \cite{li2018experimental}, whose described system collects data essentially identical to these of the present work (entries that consist of a timestamp, a MAC address and a received signal strength indicator). Their focus is on density monitoring and trajectory tracking. They do not refer to MAC address randomization, probably because their measurements were obtained a few years ago (between 2014 and 2016 according to \cite[Sec.~IV-C]{li2018experimental}), a time at which MAC address randomization was not a significant issue. Therefore, it is not clear that the accuracy of their monitoring system would be as high with today's smartphone anonymization. We reverse engineered \cite[Fig.~13]{li2018experimental} to estimate the average relative counting error and obtained a figure of 14.5 \%.\\

The authors of \cite{uras2019pma} and \cite{uras2020pma} presented a crowd monitoring based on WiFi probe requests. Their work filters out all locally administered MAC addresses \cite[Sec.~4]{uras2020pma},  relies on SHA-256 hashes without peppers \cite[Sec.~3.3]{uras2020pma} for data anonymization purposes, thereby making their anonymization procedure somewhat vulnerable to brute force attacks \cite{demir2014analysing, demir2017pitfalls, marx2018hashing}.

\begin{table}
	\centering
	\caption{Comparison of crowd counting systems most similar to that of the present manuscript --- ``Cooperation" refers to the individuals having to connect to a specific access point or install an application for the counting system to work properly --- Accuracy refers to the mean relative deviation of the counts from the ground truth (it is a mean absolute percentage error)}
	\begin{tabular}{p{1.5cm} || p{7cm} | p{2.5cm} | p{1.6cm}}
		Work(s) & Principle \& Validation & Cooperation & Accuracy \\
		\hline \hline
		\cite{denis2020large, kaya2020large, denis2020sensing} (2020) & Nodes communicating and estimating attenuation as a proxy for human presence. Validated for hundreds of attendees and more. & Not required & 14--22 \% \\
		\hline
		\cite{kaminska2019indoor} (2019) & Number of people connected to access points are measured and methods from geostatistics applied to estimate their position. Validated for $<$ 100 individuals. & Required & 4--10 \% \\
		\hline
		\cite{li2018experimental} (2018) & Density monitoring and trajectory tracking based on Wi-Fi probe requests (data set is from 2014-2016). Validated on hundreds of individuals. & Not required & $\simeq$ 14 \% \\
		\hline
		\cite{uras2019pma, uras2020pma} (2019--2020) & Density monitoring and trajectory tracking based on Wi-Fi probe requests (with randomized MAC addresses filtered out). Validated but without ground truths. & Not required & Not available \\
		\hline
		\cite{donelli2021crowd} (2021) & Crowd counting based on the analysis of the electromagnetic spectrum on cellular bands. & Not required & 5--15 \% 
	\end{tabular}
	\label{tab:comparisonWorks}
\end{table}

\subsection{Contributions} \label{subsec:contributions}

The contributions of this paper focus on a WiFi-based crowd monitoring system that detects probe requests (PRs) over the air. PRs are WiFi control packets emitted by user equipements (UEs) (e.g., smartphones) that request nearby access points (APs) to make their existence known. The rate of PR transmission is a proxy for the number of smartphones with WiFi enabled in the covered area---which, up to an \textit{extrapolation factor}, approximates the number of attendees. Thus, the extrapolation factor converts the measured rate of PRs into a number of attendees.\\

The contributions are the following:
\begin{enumerate}
	\item A novel WiFi-based sensing process enforcing strict privacy standards. This includes a time and space/memory complexity analysis and a review of privacy features.
	\item A mathematical model of the sensing process and an associated concentration inequality for the proposed unbiased crowd count estimator; it shows that it concentrates around its expectation and that the concentration increases with number of attendees.
	\item An experimental validation of the sensing process using real-world measurements from a library endowed with a third-party camera-based counting system.
\end{enumerate}

This paper relies on indoor crowd counts for experimental validation but it is merely a matter of convenience for validation by cameras: third-party camera-based counting systems can be easily installed in such controlled environments, with little need for a vast network of cameras and time-consuming calibration procedures. Installing camera systems in complex environments with numerous line-of-sight obstructions and overlapping fields of vision would be more involved.  Thereby, choosing an indoor environment with controlled entrances and exits eases the experimental validation of the counting system by providing an environment for which cameras are efficient and reliable. Nevertheless, it does not mean that the counting system cannot be installed outdoors.

\subsection{Relation to the former works of the authors}

Our previous works on forecasting \cite{determe2020forecasting, singh2020crowd}---whose main purpose was to demonstrate the interest of crowd monitoring systems for forecasting---gave a minimal overview of the counting system this manuscript presents. This manuscript details the system architecture and compares counts of the WiFi system against those from a third-party camera-based system for an indoor environment. It also presents mathematical results on the accuracy of the estimator and the effects of the anonymization procedure. Finally, it presents a detailed complexity analysis of the counting algorithm.\\

This manuscript presents new experimental results in an indoor environment. It is worth pointing out that our previous work \cite{determe2020forecasting} already provided some evidence of the accuracy of the counting system in an outdoor environment. It compared the counts generated by the counting system of this paper against those of a telecommunication operator and showed both series of counts to match. 

\subsection{Outline}

The paper is organized as follows. Section~\ref{sec:sensingProcess} describes the sensing process. In particular, it presents the mathematical model for the sensing process and the associated concentration inequality. Section~\ref{sec:digitalArchitecture} presents the digital architecture of the system, including a complexity analysis of the counting algorithm. Section~\ref{sec:legalMatters} discusses how the present system is compatible with modern European privacy laws. Section~\ref{sec:experimentalValidation} then validates the accuracy of the counting system using real-world measurements acquired at the Humanities library of our university using a third-party camera counting system. Section~\ref{sec:practicalConsiderations} briefly describes practical matters when designing and deploying WiFi monitoring systems. Finally, Section~\ref{sec:comparisonContenders} compares the present system with its contenders listed in Table~\ref{tab:comparisonWorks} and Section~\ref{sec:conclusion} is the conclusion.

\section{The sensing process} \label{sec:sensingProcess}

\subsection{The principle}

The estimated crowd counts of the counting system are derived from PRs \cite[Chapter~4]{gast2005802}. WiFi devices periodically transmit PRs to request nearby access points (APs) to send back probe responses. This is an active scanning mechanism to discover APs. WiFi devices transmit PRs even when not linked to a WiFi network. Thus, measuring the rate of PRB transmission in an area gives an idea of the number of WiFi-enabled devices in the covered area, a number which can be extrapolated to a crowd count. See Figure~\ref{fig:sensingProcess} for an illustration of the process. Several almost identical PRs are sent in a row, within a time frame lasting less than $10$ ms \cite[Sec.~2.1]{matte2016defeating}; in this paper, those sets of PRs are referred to as  probe request bursts (PRBs).

\begin{figure}[h]
	\centering
	\includegraphics[scale=3.25]{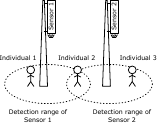}
	\caption{Two sensors sniff probe request bursts of three individuals carrying smartphones. Dashed ellipses illustrate the detection range of the associated sensor.}
	\label{fig:sensingProcess}
\end{figure}

\subsection{Probe requests}

PRs contain a source address (SA) field of six bytes \cite[Fig.~4-52]{gast2005802}, which is usually a randomized MAC address. Recent operating systems implement this randomization process to make smartphone tracking difficult \cite{freudiger2015talkative, vanhoef2016mac, matte2016defeating}. \\

Some older works from 2016-2017 show that anonymized PRs may be ``deanonymized" (see, e.g., \cite{vanhoef2016mac, martin2017study}). In the future, however, deanonymization methods may not work if operating systems strengthen anonymization. For example, \cite[Section~4]{vanhoef2016mac} partially relies on sequence numbers \cite[Figure~4-52]{gast2005802}, which are numbers associated with each PR that are incremented in between consecutive PRs. So far, it appears that such sequence numbers are not randomly reset from one PRB to the next one---a fact that the authors \cite{gast2005802} leverage to track smartphones. Should sequence numbers be randomly reset in the future, the strategy may not work anymore. More generally, MAC address randomization is likely to get tougher in the future \cite[Sec.~1]{denis2020large}; as pointed out in \cite[Sec.~4]{uras2020pma}, ``the IEEE 802.11 working group has created a Topic Interest Group (TIG) on Randomized and Changing MAC addresses (RCM)".\\

As discussed later on, MAC address randomization does not affect the counting system, which makes it future-proof, in opposition to other WiFi monitoring systems either deanonymizing PRBs or identifying non-randomized PRBs (see \cite{uras2020pma}).

\subsection{A mathematical sensing model} \label{subsec:mathSensingModel}

This section derives the statistical estimator that estimates counts from a measured rate of PRB transmissions. It also presents a statistical analysis of the estimator,  deriving its distribution and a concentration inequality for it. In what follows, $\mathbb{P}$ and $\mathbb{E}$ denote the probability of an event and the mathematical expectation, respectively.\\

First of all, let $n_{\rm ppl}$ denote the number of individuals in an area. This is the quantity the estimator should estimate as accurately as possible. In what follows, index $i$ ($1 \leq i \leq n_{\rm ppl}$) indexes a particular individual.\\

These individuals may or may not have a device with features enabled. Moreover, smartphones send PRBs at different rates depending on the operating system version. These two effects are accounted for by random variables $P_i$ ($1 \leq i \leq n_{\rm ppl}$) defined for each individual: $P_i$ is the average number of PRBs with different source addresses that the WiFi device carried by individual $i$ sends over the air per time frame of $t_f$ seconds. The time frame duration $t_f$ is assumed to be sufficiently short to ensure that no WiFi device sends PRBs with different source addresses in that time frame; as a result, $P_i \leq 1$. If individual $i$ carries no WiFi-enabled device or has disabled WiFi on a capable device, then $P_i = 0$. The $P_i$  are independently and identically distributed (iid.).  We denote the mean of $P_i$ by $p$, that is $\mathbb{E} \lbrack P_i \rbrack =: p$.\\

There are $K < \infty$ possible values $\lbrace \alpha_k \rbrace_{1 \leq k \leq K}$ for $P_i$ because there exists a finite number operating system configurations; the probability $r_k := \mathbb{P} \lbrack p_i = \alpha_k \rbrack$ obeys $\sum_{k=1}^K r_k = 1$, with $\alpha_k = 0$ corresponding to individual $i$ having no WiFi-enabled device.\\

The number of distinct PRBs within a time frame of $t_f$ seconds is $X := \sum_{i=1}^{n_{\rm ppl}} X_i$, with $X_i$ being equal to $1$ if individual $i$'s WiFi device sends a PRB. The equalities $\mathbb{P}[X_i = 1 | P_i = \alpha_k] = \alpha_k$ and $\mathbb{P}[X_i = 0 | P_i = \alpha_k] = 1 - \alpha_k$ follow from this definition. Hence, the law of total probability shows that the marginal distribution of $X_i$ obeys \cite[Sec.~II-D]{determe2020forecasting} $\mathbb{P}[X_i = 1] = \sum_{k=1}^K \mathbb{P}[X_i = 1 | P_i = \alpha_k] \mathbb{P}[P_i = \alpha_k] = \sum_{k=1}^K \alpha_k r_k =: \mathbb{E} \lbrack P_i \rbrack =: p$. The mean of $X_i$ is $\mathbb{E} \lbrack X_i \rbrack := 1 \; \mathbb{P} \lbrack X_i = 1 \rbrack + 0 \; \mathbb{P} \lbrack X_i = 0 \rbrack = \mathbb{E} \lbrack P_i \rbrack =: p$. Consequently, an unbiased estimator of the number of individuals $n_{\rm ppl}$ is 
\begin{equation} \label{eq:ChatBasicDef}
	\hat{C} := \beta X,
\end{equation}
where $\mathbb{E} \lbrack \hat{C} \rbrack = n_{\rm ppl}$ with extrapolation factor $\beta := 1/p$. Variable $X$ is a sum of $n_{\rm ppl}$ statistically independent and identically distributed (iid) Bernoulli random variables $X_i$ of parameter $p := \mathbb{E} \lbrack P_i \rbrack$. As a result, $X = \hat{C}/\beta$ follows a binomial distribution $B(n_{\rm ppl},p)$.

\subsection{Concentration inequalities and asymptotic analysis}

Now that an unbiased estimator  and its distribution have been derived, this subsection derives a a concentration inequality for the estimator $\hat{C}$ around its mean. Loosely speaking, this inequality is theoretical evidence that the estimator is reliable. Results from \cite{buldygin2013sub} are used and the resulting concentration inequality is compared against a canonical concentration inequality for bounded random variables. A key quantity depending on $p$ is $K(p)$, defined below.
\begin{amsdef} \label{amsdef:defKp}
	Let $K : \lbrack 0, 1 \rbrack \rightarrow \mathbb{R} : p \mapsto K(p)$, where \cite[Eq.~(4)]{buldygin2013sub}
	\begin{equation}\label{eq:defKp}
		K(p) = 
		\begin{cases}
			0 & \textrm{ if } p \in \lbrace 0, 1 \rbrace\\
			1/4 & \textrm{ if } p = 1/2 \\
			\dfrac{p-q}{2(\log p - \log q)} & \textrm{ if } p \in (0,1) \backslash \lbrace 1/2 \rbrace
		\end{cases},
	\end{equation}
	with $q := 1 - p$.
\end{amsdef}
Proposition~\ref{amsprop:KpProperties} helps understanding the shape of $K(p)$.
\begin{amsprop} \label{amsprop:KpProperties}
	With $K$ defined as in~(\ref{eq:defKp}), we have the following properties:
	\begin{enumerate}
		\item $K$ is continuous and convex
		\item $K$ is symmetric around $p = 1/2$
		\item $K$ increases on $p \in \lbrack 0, 1/2 \rbrack$ and decreases on $\lbrack 1/2, 1 \rbrack$
		\item $K(p) \leq 1/4$.
	\end{enumerate}
\end{amsprop}
\begin{proof}
	All statements are available almost verbatim in \cite[Lemma~2.1]{buldygin2013sub}.
\end{proof}
Proposition~\ref{amsprop:concIneq} states the concentration inequality for $\hat{C}$.
\begin{amsprop} \label{amsprop:concIneq}
	With $K$ defined by~(\ref{eq:defKp}) and $\hat{C}$ by~(\ref{eq:ChatBasicDef}), we have, for any $\varphi > 0$, 
	\begin{equation} \label{eq:concIneq}
		\mathbb{P} \lbrack | \hat{C} - n_{\rm ppl} | \geq \varphi n_{\rm ppl} \rbrack \leq 2 \exp \left( - \dfrac{\varphi^2}{2} n_{\rm ppl} \dfrac{p^2}{K(p)} \right).
	\end{equation}
\end{amsprop}
\begin{proof}
	The previous subsection has already shown that $X$ is of a sum of $n_{\rm ppl}$ iid. Bernoulli random variables of parameter $p$. Thus, \cite[Corollary~6.1 (ii)]{buldygin2013sub} directly implies
	\begin{equation*}
		\mathbb{P} \lbrack |X - n_{\rm ppl} p| \geq x \rbrack \leq 2 \exp \left( \dfrac{-x^2}{2 n_{\rm ppl} K(p)} \right).
	\end{equation*}
	With $\hat{C} = \beta X = X/p$ and $x = \varphi n_{\rm ppl}p$,
	\begin{align*}
		\mathbb{P} \lbrack |\hat{C} - n_{\rm ppl}| \geq \varphi n_{\rm ppl} \rbrack & = \mathbb{P} \lbrack |X - n_{\rm ppl} p| \geq \varphi n_{\rm ppl}p \rbrack \\
		& \leq 2 \exp \left( - \dfrac{\varphi^2}{2} n_{\rm ppl} \dfrac{p^2}{K(p)} \right). \qedhere
	\end{align*}
\end{proof}
This concentration inequality upper bounds the probability that $\hat{C}$ diverges from its mean $n_{\rm ppl}$ as a function of a proportion $\varphi$ of the mean. In particular, it shows that the probability of a divergence of $\varphi n_{\rm ppl}$ decreases exponentially with the number of people in the area $n_{\rm ppl}$---which means that the relative accuracy of the estimator increases with $n_{\rm ppl}$ and becomes infinite as $n_{\rm ppl} \rightarrow \infty$.\\

As $\lim_{p \rightarrow 1^-} p^2 K(p)^{-1} = + \infty$, if every individual is guaranteed to send one PRB ($p=1$), the relative estimator accuracy is infinite. Conversely, using L'Hôpital's rule,
\begin{equation*}
	\lim_{p \rightarrow 0^+} \dfrac{p^2}{K(p)} = \lim_{p \rightarrow 0^+} \dfrac{2 \log p^{-1}}{p^{-2}} = \lim_{p \rightarrow 0^+} \dfrac{- 2 p p^{-2} }{-2 p^{-3}} = 0,
\end{equation*}
which suggests that if no individual sends PRBs ($p=0$), the estimator is worthless. \\

Finally, this section compares~(\ref{eq:concIneq}) against Hoeffding's inequality (see \cite{hoeffding1963probability} and \cite[Theorem~2.8]{boucheron2013concentration}). Without proof details, one easily obtains Hoeffding's inequality: 
\begin{equation} \label{eq:hoeffdingsIneq}
	\mathbb{P} \lbrack | \hat{C} - n_{\rm ppl} | \geq \varphi n_{\rm ppl} \rbrack \leq 2 \exp \left( - 2\varphi^2 n_{\rm ppl} p^2 \right),
\end{equation}
which is also obtained by using~(\ref{eq:concIneq}) and $K(p) \leq 1/4$ (see \cite[Remark~5.1]{buldygin2013sub}), which shows~(\ref{eq:concIneq}) outperforms~(\ref{eq:hoeffdingsIneq}). In particular, the Hoeffding's inequality fails to predict that $p = 1$ implies a perfect accuracy of the estimator, a task at which the presented concentration inequality~(\ref{eq:concIneq}) succeeds.

\section{Digital architecture} \label{sec:digitalArchitecture}

\subsection{Overview}

Our system comprises i) a set of sensors, ii) a processing subsystem on a central server collecting all PRBs and processing them in real time, and iii) a dumping subsystem (that is part of the central server) that further anonymizes and then dumps PRBs. All communications between the sensors and the central server use layers of authentication; they are secured using HTTPS, thereby encrypting packets and also preventing man-in-the-middle attacks. Figure~\ref{fig:GenArchi} depicts the general system architecture, with each of the three subsystems discussed in the next subsections.

\begin{figure*}[h]
	\centering
	\includegraphics[scale=1.85]{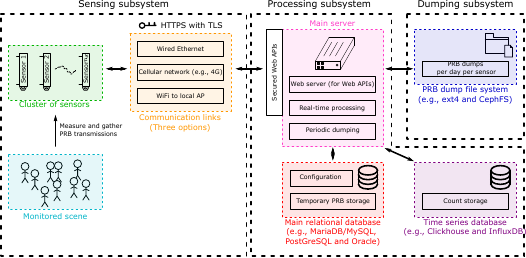}
	\caption{General architecture of the counting system}
	\label{fig:GenArchi}
\end{figure*}

\subsection{The sensing subsystem} \label{subsec:sensingSubsystem}

As shown in Figure~\ref{fig:GenArchi}, the sensing subsystem may be decomposed in three parts: the scene for which to estimate crowd counts, the cluster of $n_{\mathcal{S}}$ sensors deployed to count the crowd and a communication link for each sensor. The communication links may be a wired Ethernet connection, a cellular link or a Wi-Fi connection to a local access point (AP), and they may be different among sensors. Although all three link options are viable, the experiments this manuscript describes were made using 4G communication links only.\\

The sensors i) detect PRBs, ii) anonymize them and iii) send them to a central server.

\subsubsection{Hardware} Each WiFi sensors comprises \cite[Sec.~II-B]{determe2020forecasting}
\begin{itemize}
	\item A Raspberry Pi 3B (running Raspbian Stretch).
	\item An \textit{Alfa AWUS036NHA} WiFi dongle (chipset \textit{Atheros AR9271L}) supporting monitor mode---a state that makes the dongle capture all over-the-air WiFi messages, without being restricted to those of a particular WiFi network. The dipole antennas shipped with \textit{Alfa AWUS036NHA} dongles equip sensors. Sensor antennas point perpendicularly to the ground.
	\item A 4G dongle granting access to the Internet.
\end{itemize}
 
Figure~\ref{fig:figBox} shows a photograph of a sensor.

\begin{figure}[h]
	\centering
	\includegraphics[scale=0.4]{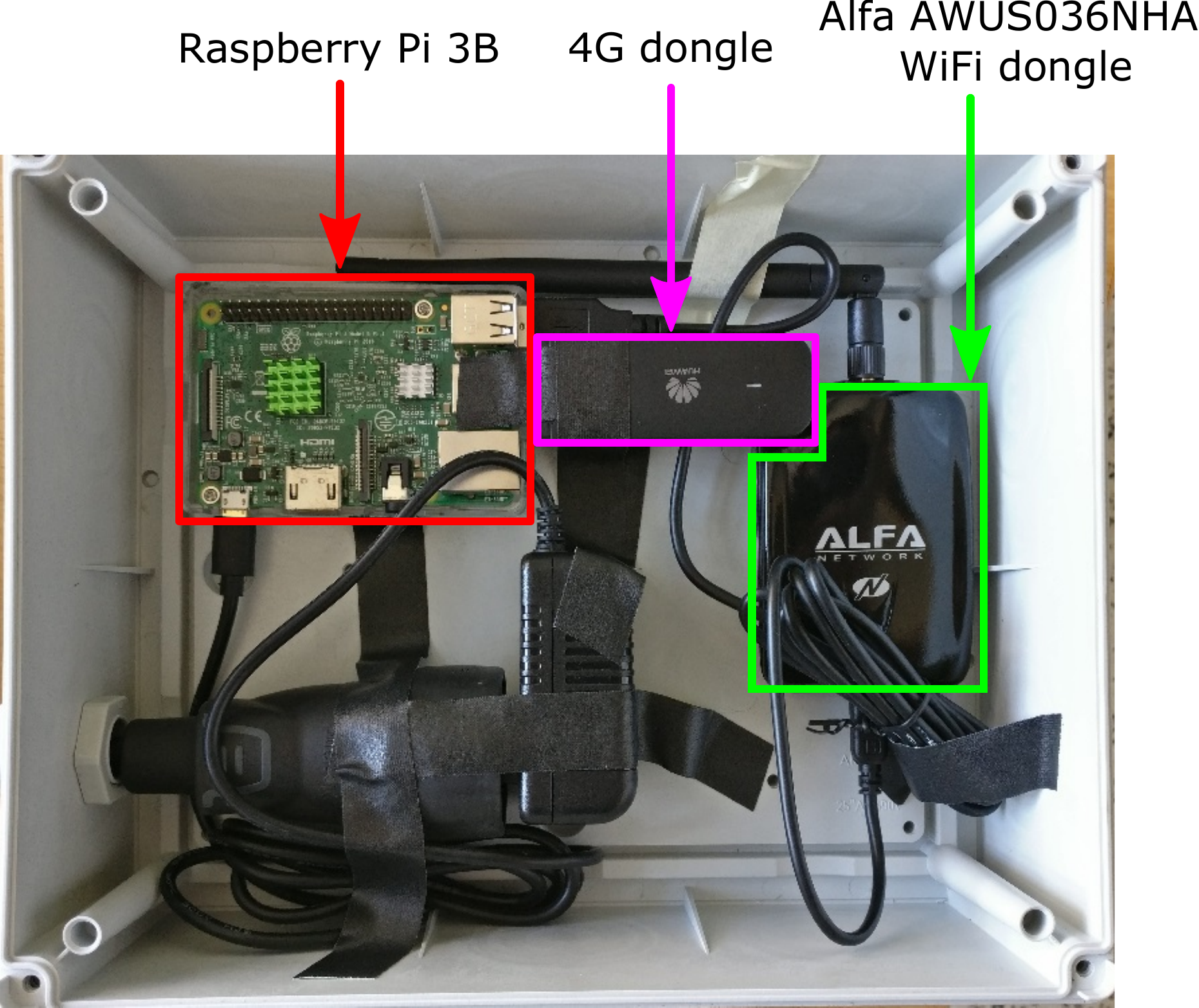}
	\caption{Photograph of the inside of a sensor.}
	\label{fig:figBox}
\end{figure}

\subsubsection{Software}
The sniffing program has been written in multi-threaded C++ and uses packet capture library \textit{libpcap}. For each detected PRB, sensors send \cite[Sec.~II-B]{determe2020forecasting} ``\textit{i) an anonymized MAC address of the PRB, ii) the timestamp of detection iii) a received signal strength indicator (RSSI) value, which is a number quantifying the received power}". Stress tests of the sensors have revealed that neither the WiFi dongle nor the Raspberry Pi fail to handle large PRBs transmission rates.
\subsubsection{Anonymization}

All sensors periodically retrieve from the central server an up-to-date array of (cryptographic) \textit{server peppers}. Each pepper of the array is associated with a one-minute time frame, during which it will be used. The central server regenerates the server peppers in real time, and it deletes old peppers so that they cannot be retrieved in the future. The server uses an entropy pool (\texttt{/dev/urandom} on Linux distributions) to generate cryptographically secure peppers. A \textit{sensor pepper} is also hardcoded in the C++ codebase of all sensors; it is common to all sensors (at least all sensors located in the same area and thus likely to detect identical PRBs simultaneously). It is a final line of defense in case the server peppers get compromised.\\

As depicted in Figure~\ref{fig:anonymProcedureSensor}, for every received probe request, the sensor prepends a global pepper to the full MAC address before computing the SHA-256 hash of the concatenated byte sequence, whose 256 bits are truncated to 64 bits. The pepper is the concatenation of the sensor pepper and the server pepper, both of 128 bits. \\

\begin{figure}[h]
	\centering
	\includegraphics[scale=0.7]{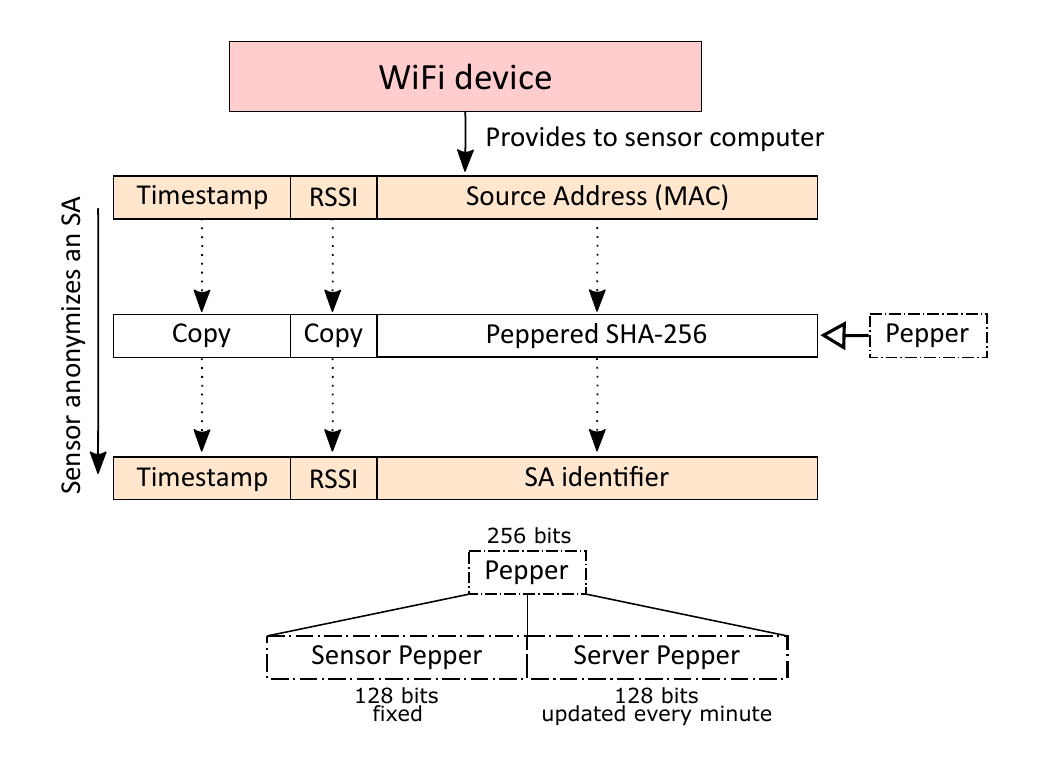}
	\caption{(From \cite{determe2020collisions}) Scheme of the anonymization procedure executed by sensors}
	\label{fig:anonymProcedureSensor}
\end{figure}

As shown in \cite{determe2020collisions}, the system meets four essential requirements. First, time synchronization is accurate enough to make sensors use identical peppers at identical time instants (at least when operating on networks with low latency, such as LTE networks \cite{mivskinis2014timing}). Second, from the SA identifiers, it is realistically impossible to recover the original MAC addresses. Third, tracking individuals for more than one minute is not possible. Fourth, the collision rate of the truncated SHA-256 hash is less than $10^{-9}$ for $10^7$ MAC addresses (which corresponds to an unrealistically high number of individuals). Satisfying the first and fourth requirements ensures anonymization does not tamper with the counting method. The second and fourth requirements consist in privacy-enhancing features.

\subsection{The processing subsystem} \label{subsec:processingSubsystem}
The processing subsystem of Figure~\ref{fig:GenArchi} comprises three submodules. The first one, referred to as ``Web server" is there to allow sensors to interact with the server through secured Web APIs. The second, ``Real-time processing" is the process computing counts, a process that is extensively detailed in what follows. The third, periodic dumping, triggers a dump of PRBs temporarily stored in the main relational database into the PRB dump file system. This remaining part of this subsection discusses the ``Real-time processing'' submodule.\\

The PRBs measured by all sensors are to be processed jointly and usually in real time (this corresponds to ``Real-time Processing" in the processing subsystem of Figure~\ref{fig:GenArchi}). The task here is to generate a count for each time frame of one minute and each sensor of an event, while counting smartphones detected simultaneously by several sensors only once. This will be accomplished by looping through each time frame of one minute and, for each one of them, two main steps are to be carried out: \textit{i)} a filtering operation that extracts all PRBs measured during the time frame \textit{ii)} the association of each observed anonymized MAC address in the filtered dataset to only one sensor: the one having measured the highest signal power---this is a coarse measure of proximity between the device transmitting the PRB and each sensor.\\

As shown in Figure~\ref{fig:GenArchi}, PRBs are stored in a typical relational database or in a file system hosting binary files (each of which gathers PRBs for a specific sensor ID and 24-hour period). Every PRB of the dataset consists of four entries:
\begin{enumerate}
	\item A timestamp \texttt{ts} (whose precision is of one second) that indicates when the PRB has been acquired
	\item A sensor ID \texttt{sensorid} indicating which sensor acquired the PRB
	\item An anonynomized MAC address \texttt{amac} of 64 bits
	\item A received signal strength indicator (RSSI) \texttt{rssi} that quantifies the received power when detecting the PRB.
\end{enumerate} 
In a relational database, an index allows for efficient search using the timestamp \texttt{ts} whereas, in a file system, all files store PRBs sorted by their timestamps.

\subsubsection{First stage of the counting algorithm}
The starting point of processing PRBs is about extracting all the PRBs that have been measured within a one-minute time frame (e.g., from 11:29:01 AM to 11:30:00 AM). This means the 4-tuples (\texttt{ts}, \texttt{sensorid}, \texttt{amac}, \texttt{rssi}) from the database go through a filter that only keeps the entries for which \texttt{ts} is within the time frame limits. This is an easy task because PRBs are already indexed or sorted by their timestamps. This first operation provides a reduced dataset of 3-tuples (\texttt{sensorid}, \texttt{amac}, \texttt{rssi}) that is one of the inputs of the second stage.

\subsubsection{Second stage of the counting algorithm}

Algorithm~\ref{alg:prbprocessing} describes the second stage. Besides the reduced dataset from the first stage, which is the array \texttt{arr\_mac}, the algorithm also uses as input a user-provided hash table of RSSI lower bounds for each sensor, whose key is a sensorid and whose value is an object with only one field, rssilowerbound. This lower bound allows users to exclude any PRB measured by a given sensor whose RSSI value is below rssilowerbound. Because the RSSI is linked to the distance to the sensor, it provides a coarse way of tuning the effective detection range of sensors. Such RSSI bounds are typically stored in the relational database in Figure~\ref{fig:GenArchi} under the name ``Configuration".\\

Besides the inputs, the algorithm initializes an empty hash table \texttt{ht} whose keys are anonymized MAC addresses \texttt{amac} and values are a 2-tuples (\texttt{sensorid}, \texttt{RSSI}), see step 1 in Algorithm~\ref{alg:prbprocessing}). It keeps track of the highest measured RSSI for each anonymized MAC address and of the sensor having measured it. Algorithm~\ref{alg:prbprocessing} also initializes an array of counts \texttt{counts\_per\_sensor} that is initially filled with zeroes (step 17 in Algorithm~\ref{alg:prbprocessing}) and will eventually contain the counts for each sensor for the time frame being processed.\\

The algorithm loops through every reduced PRB in \texttt{arr\_mac} (a 3-tuple (\texttt{sensorid}, \texttt{amac}, \texttt{rssi}) denoted by \texttt{prb}) and extracts its sensor ID (\texttt{sensorid}) and its anonymized MAC address (\texttt{amac}), see steps 2 to 4 in Algorithm~\ref{alg:prbprocessing}. It then determines whether the PRB is to be discarded immediately because its RSSI is below the prescribed threshold for the sensor (step 5). If not discarded, it checks whether the anonymized MAC address \texttt{amac} has already been encountered before (step 6). If so and if the RSSI measured \texttt{prb}.rssi is higher than those encountered so far for \texttt{amac} (step 7), \texttt{ht}[\texttt{amac}] is modified so that \texttt{sensorid} becomes the sensor ID for which the highest RSSI has been observed for \texttt{amac} (steps 8 and 9). Similarly, if \texttt{amac} has never been observed before (step 11), \texttt{ht}[\texttt{amac}] is modified identically (steps 12 and 13).\\

At step 17 of Algorithm~\ref{alg:prbprocessing}, \texttt{ht} contains all the observed anonymized MAC addresses (without duplicates) and, for each one of them, it provides the sensor ID having measured the highest RSSI. It is then sufficient to perform steps 18 to 20 to compute the number of unique devices estimated to be the closest to each sensor. A final step before returning the sensor counts for the current one-minute time frame is step 21, which exists to be explicit about the cleaning of \texttt{ht} and its impact on time complexity.\\

In practice, the counts obtained are stored in a specialized database for time series (see Figure~\ref{fig:GenArchi}). InfluxDB is an example and, with the right compression codecs, Clickhouse has also proved to provide compact storage as well as fast querying. Both databases can be distributed across several nodes to offer robustness, scalability and high throughputs.

\begin{figure}[!h]
	\textsc{Algorithm \refstepcounter{algoCounter}\label{alg:prbprocessing}\arabic{algoCounter}}:\\ 
	Compute counts for a single time frame from PRBs. Comments on the right indicate time complexity (steps without complexity have time complexity $\mathcal{O}(1)$).\\
	\vspace{-2mm}
	%\begin{algorithmic}[1]
		\begin{boxedalgorithmic}
			\small
			\Require List of PRBs \texttt{arr\_mac} for the time frame of interest only, 3-tuples (sensorid, amac, rssi); Hash table \texttt{sensors}, of key sensorid and of value $\lbrace$rssilowerbound$\rbrace$
			\State Initialize a hash table \texttt{ht} whose keys are 64-bit anonymized MAC addresses (or equivalently, random tokens) and whose values are 2-tuples (sensorid, rssi).
			\LineComment{Loop through all PRBs associated to time frame of interest}
			\ForAll {\texttt{prb} \textbf{in} \texttt{arr\_mac}} \Comment{$\mathcal{O}(\text{length of {\ttfamily arr\_mac}})$}
			\State \texttt{sensorid} := \texttt{prb}.sensorid \Comment{ $\mathcal{O}(1)$}
			\State \texttt{amac} := \texttt{prb}[amac] \Comment{$\mathcal{O}(1)$}
			\LineComment{Check if PRB should be discarded based on RSSI}
			\If {\texttt{prb}.rssi $>$ \texttt{sensors}[\texttt{sensorid}].rssilowerbound}
			\LineCommentIndent{Check if \texttt{amac} already detected previously}
			\If {\texttt{amac} \textbf{in} \texttt{ht}} \Comment{$\mathcal{O}(1)$}
			\LineCommentIndent{Check if PRB detected has highest RSSI}
			\LineCommentIndent{for \texttt{amac} among those of all sensors}
			\If {\texttt{prb}.rssi $>$ \texttt{ht}[\texttt{amac}].rssi}
			\LineCommentIndent{A new highest RSSI has been found}
			\State \texttt{ht}[\texttt{amac}].sensorid := \texttt{sensorid}
			\State \texttt{ht}[\texttt{amac}].rssi := \texttt{prb}.rssi
			\EndIf
			\Else
			\LineCommentIndent{\texttt{amac} detected for the first time}
			\State \texttt{ht}[\texttt{amac}].sensorid := \texttt{sensorid}
			\State \texttt{ht}[\texttt{amac}].sensorid := \texttt{prb}.rssi
			\EndIf
			\EndIf
			\EndFor
			\State Initialize array of counts \texttt{counts\_per\_sensor} with zeros
			\ForAll {\texttt{elem} \textbf{in} \texttt{ht}} \Comment{$\mathcal{O}(\text{length of {\ttfamily ht}})$}
			\State \texttt{counts\_per\_sensor}[\texttt{elem}.sensorid] += 1 \Comment{$\mathcal{O}(1)$}
			\EndFor
			\State Empty hash table \texttt{ht} \Comment{$\mathcal{O}(\text{length of {\ttfamily ht}})$}
			\State \Return \texttt{counts\_per\_sensor}
			%\end{algorithmic}
	\end{boxedalgorithmic}
	%\caption{Simultaneous orthogonal matching pursuit (SOMP) algorithm}
\end{figure}

\subsubsection{Complexity analysis}

\begin{figure}
	\begin{lstlisting}[style=CStyle]
		struct prb
		{
			time_t ts; // 32-bit UNIX timestamp
			uint16_t sensorid; // Sensor ID
			int8_t aMAC[8]; // Anonymized MAC addr.
			int8_t rssi; // RSSI
		};
	\end{lstlisting}
	\caption{Example of a C structure representing a probe request burst. In this case, any standard compiler appends 1 trailing pad byte for data alignment purposes; thus, the structure size is 16 bytes. The size of \textit{rssi} is that of the \textit{antenna signal} field of standard \textit{RadioTap} headers.}
	\label{fig:prbCStruct}
\end{figure}

This subsection deals with complexity analysis (in time and in space). Let $n_{\mathcal{S}}$ denote the number of sensors, each one of which capturing no more than $n_{\rm meas}$ PRBs for a one-minute time frame.With the data structure in Figure~\ref{fig:prbCStruct}, storing all the PRBs for a given time frame has a memory footprint of $n_{\mathcal{S}} n_{\rm meas}\; 16\; 10^{-6}$ MB. \\

The memory of the hash table \texttt{ht} used for processing PRBs is also reasonable. Let $n_b$ denote the number of buckets of the hash table. In practice, $n_b$ can be chosen to get a load factor lower than or equal to $\alpha$ so that $n_b = n_{\mathcal{S}} n_{\rm meas} \alpha^{-1}$. Setting the number of buckets beforehand requires one to know the maximum number of PRBs per time frame attained in practice (and the proportion of duplicated PRBs).\\

With a C structure similar to that of Figure~\ref{fig:prbCStruct}, each 2-tuple (\texttt{sensorid}, \texttt{rssi}) of the hash table consists of 8 bytes (including two trailing pad bytes). Assuming that collision resolution relies on separate chaining with linked lists \cite[Chap.~11]{cormen2009introduction}, the baseline memory footprint of the hash table is equal to $n_b \; 8 \; 10^{-6}$ MB on a 64-bit architecture. Every node of the linked list has a memory footprint of $16$ bytes ($8$ bytes for the pointer and $8$ bytes for the 2-tuple value). Thus, loading $n_{\mathcal{S}} n_{\rm meas} \alpha^{-1}$ entries in the hash table has a memory footprint of $n_{\mathcal{S}} n_{\rm meas} (8 \; \alpha^{-1} + 16) \; 10^{-6} $ MB (the first and second terms correspond to the bucket pointers and the nodes of the linked lists, respectively).\\

As a conclusion, processing a large event is computationally tractable from a space complexity point of view. For large events lasting several days, loading all the measurements in memory at once may be impossible but is also pointless: the proposed method processes time frames sequentially and independently from one another.\\

It is now time to turn to time complexity. With a properly designed hash table, insert and search operations have an average time complexity of $\mathcal{O}(1)$. Looping through all entries in \texttt{arr\_mac} has a time complexity of $\mathcal{O}(n_{\mathcal{S}} n_{\rm meas})$. The reason is that the number of loops is $n_{\mathcal{S}} n_{\rm meas}$ (step 2), each one of which including only operations of time complexity  $\mathcal{O}(1)$. Counting all entries in the final hash table with specific sensor IDs has a time complexity of $\mathcal{O}(n_{\mathcal{S}} n_{\rm meas})$ because the prescribed load factor makes the number of buckets directly proportional to $n_{\mathcal{S}} n_{\rm meas}$. Releasing the linked lists of all buckets also has a time complexity of $\mathcal{O}(n_{\mathcal{S}} n_{\rm meas})$ (step 21). Globally, the average time complexity is $\mathcal{O}(n_{\mathcal{S}} n_{\rm meas})$. It is easy to show that the worst-case time complexity is $\mathcal{O}((n_{\mathcal{S}} n_{\rm meas})^2)$---which is attained if all the SA identifiers are mapped onto the same bucket, thereby creating a unique linked list of size $n_{\mathcal{S}} n_{\rm meas}$.

\subsection{The dumping subsystem} \label{subsec:dumpingSubsystem}

\subsubsection{Principle}
The system periodically dumps PRBs from the SQL database into binary files stored in the ``PRB dump file system" in Figure~\ref{fig:GenArchi}. Each dump file corresponds to a particular sensor and a particular day. This keeps in check the size of the SQL table storing PRBs and its indexes. It also makes it straightforward to backup these files in a cheap storage location (e.g., in "cold storage`` facilities).\\

If the system does not ingest excessive throughputs of data, storing binary dump files on ext4 file systems is acceptable and can easily support volumes of at least 8 terabytes (using conventional hard drives or storage solutions from cloud providers). Otherwise, is it possible to use a distributed file system such as CephFS; the latter option provides redundancy, scalable IO throughputs and support for volumes larger than 10 petabytes.

\subsubsection{Anonymization}
The SQL database stores anonymized MAC addresses; theoretically, a deterministic link still exists between the original MAC address and its corresponding SA identifier. Removing the link is beneficial because someone could identify a vulnerability of SHA256 in the future. Therefore, the dumping program randomizes SA identifiers per time frame using, e.g., the Mersenne twister. The links ``SA identifier $\rightarrow$ final SA identifier" are reset after each time frame of one minute. A cryptographically secure pseudorandom number generator (CSPRNG) is not needed as the only requirements are i) to remove any deterministic link between the original MAC address and the identifier ii) and having uniformly distributed identifiers. This approach also makes it impossible for hackers to revert their way back to the original SAs on the basis of the dump files, even if they intercept the peppers.\\

\section{Legal matters about privacy} \label{sec:legalMatters}
Nowadays, an important topic about crowd monitoring systems is whether they comply with privacy laws. This is particularly true in Europe since May 25, 2018---the date that saw the advent of the European general data protection regulation (GDPR). The present system satisfies European and Belgian privacy laws because it does not allow administrators or third parties to (see Section~\ref{subsec:sensingSubsystem})
\begin{itemize}
	\item recover the original MAC addresses or other personal data about individuals carrying the detected WiFi devices,
	\item track MAC addresses over time.
\end{itemize}
In this sense, it is possible to consider that no personal data are processed and, as a result, tracked individuals need not be informed of tracking.

\section{Experimental validation} \label{sec:experimentalValidation}
A previous experimental evaluation focusing on the extrapolation factor for public events is \cite[Sec.~II-E and Fig.~2]{determe2020forecasting}; this former analysis compares the WiFi counts the system generates with those from a telco operator.  This paper and section provide an experimental evaluation of the accuracy of the WiFi system in an indoor environment. Experimental validation relies on third-party counts from Affluences and their \textit{3D Video sensor} system \cite{affluences2020camera}, which has been installed at the entries and exits of the Humanities library at Université libre de Bruxelles (ULB). This provides a ground truth from a third-party, commercially available counting system.\\

As in \cite{determe2020forecasting}, two accuracy measures are used: the root mean square error (RMSE) and the mean absolute percentage error (MAPE). For a time series $\lbrace x_t \rbrace_{0\leq t \leq N-1}$ of $N$ true counts and a time series $\lbrace \hat{x}_t \rbrace_{0\leq t \leq N-1}$ of $N$ approximated counts,

\begin{equation} \label{eq:defRMSE}
	\mathrm{RMSE} :=  \sqrt{\dfrac{1}{N} \sum_{t=0}^{N-1} (x_t - \hat{x}_t)^2}
\end{equation}
and
\begin{equation}
	\mathrm{MAPE} := \dfrac{100 \%}{N} \sum_{t=0}^{N-1} \dfrac{|x_t - \hat{x}_t|}{|x_t|}.
\end{equation}

Both accuracy measures are extensively used in the literature. RMSE is an absolute measure of the error variance and thus tends to penalize high errors proportionally more than smaller ones because of its quadratic nature. MAPE is a relative measure of the error $x_t - \hat{x}_t$ normalized using the ground truth time series $\lbrace x_t \rbrace_{0\leq t \leq N-1}$. MAPE penalizes errors linearly but an absolute error tends to be penalized more if it is associated to a low ground truth count.

\subsection{Measurement setup}

The measurement setup at the Humanities library consists in six sensors installed on three (consecutive) floors of an eight-story building.

\subsection{Extrapolation to account for partial coverage}

In ideal circumstances, sensors cover the whole area to be monitored. In practice, budget or infrastructure constraints may prevent an installation with full coverage and the total counts of people are extrapolated on the basis of counts for a sub-area. Thereby, with $\hat{C}^{\rm (part)}$ denoting the (partial) counts for the covered sub-area, 
\begin{equation} \label{eq:kappaBetaExtrap}
	\hat{C} = \kappa \hat{C}^{\rm (part)} = \kappa \beta X,
\end{equation}
where $\hat{C}$, $\beta$ and $X$ are defined in Section~\ref{subsec:mathSensingModel} and $\kappa$ is an extrapolation factor converting counts for the sub-area into counts for the whole area. (If the whole area is covered, $\kappa = 1$.) The global extrapolation factor is then $\tilde{\beta} := \kappa \beta$. A more complete model that includes noise signals for both extrapolations is
\begin{align*}
	\hat{C} & = \kappa \hat{C}^{\rm (part)} + e^{(\kappa)} \\
	& = \kappa (\beta \dfrac{C}{\kappa \beta} + \epsilon^{(\beta)}) + e^{(\kappa)} \\
	& = C + \kappa \epsilon^{(\beta)} + e^{(\kappa)}\ ,
\end{align*}
where $C$ denotes the true count whereas $e^{(\kappa)}$ and $\epsilon^{(\beta)}$ denote errors linked to the two extrapolation procedures. As the Affluences cameras provide counts for the whole library and the WiFi system covers three stories of out eight, $\kappa > 1$ and $e^{(\kappa)} \neq 0$.

\subsection{Estimate the global extrapolation factor} \label{subsec:estGlobalExtrapFactor}

The global extrapolation factor $\tilde{\beta}$ shall be fit using a least squares approach with $N$ measurements for each subsystem. Let $\boldsymbol{c}^{\rm Affl.} \in \mathbb{R}^N$ and $\boldsymbol{c}^{\rm WiFi} \in \mathbb{R}^N$ denote counts from the Affluences cameras and WiFi subsystem, respectively.Affluences/Camera and WiFi counts are available every 30 minutes and 5 minutes, respectively. The WiFi count series is thus downsampled by 6 to obtain comparable and compatible time series for both subsystems. The linear model $\bsy{y} = \bsy{A} \bsy{x}$ is particularized by the substitutions $\bsy{y} = \boldsymbol{c}^{\rm Affl.}$, $\bsy{A} = \boldsymbol{c}^{\rm WiFi}$ and $\bsy{x} = \lbrack \tilde{\beta} \rbrack \in \mathbb{R}^{1 \times 1}$. The pseudoinverse of $\bsy{A}$ with linearly independent columns is $\bsy{A}^+ = (\boldsymbol{c}^{\rm WiFi})^+ = ((\boldsymbol{c}^{\rm WiFi})^{\rm H} \boldsymbol{c}^{\rm WiFi})^{-1} (\boldsymbol{c}^{\rm WiFi})^{\rm H}$, which provides a least squares estimate for $\tilde{\beta}$ that is 
\begin{equation} \label{eq:globalEqFactorEstimate}
	\textrm{Estimate} \lbrack \tilde{\beta} \rbrack := \langle \boldsymbol{c}^{\rm WiFi}, \boldsymbol{c}^{\rm Affl.} \rangle / \| \boldsymbol{c}^{\rm WiFi} \|_2^2,
\end{equation}
where $\langle \boldsymbol{c}^{\rm WiFi}, \boldsymbol{c}^{\rm Affl.} \rangle$ denotes the inner product of $\boldsymbol{c}^{\rm WiFi}$ and $\boldsymbol{c}^{\rm Affl.}$.

\subsection{Preprocessing pipeline}

The WiFi system works all the time; however, its accuracy should only be evaluated during opening times. To do this, a preprocessing pipeline processes both the Affluences and WiFi time series in the following way:
\begin{enumerate}
	\item Extract a particular time frame with counts available every 30 minutes (all days from 2019-04-02 until 2019-06-01)
	\item Remove week-ends, holidays and days during which any of the two systems was malfunctioning; the following days were removed: 2019-04-22 (holiday), 2019-05-01 (holiday), 2019-05-14 (tests), 2019-05-23 (Affluences malfunction) and 2019-05-30 (holiday).
	\item Restrict the time ranges to those during which the library is guaranteed to be opened (from 9:00 AM to 6:00 PM).
\end{enumerate}

\subsection{Results with a unique global extrapolation factor} \label{subsec:resultsGlobalExtrapFactor}

As a first step, the analysis relies on the pessimistic assumption that $\tilde{\beta}$ is constant over time. This is not necessarily true because sensors cover only three floors of the library and students pursue different endeavors over time; for example, many projects are over by May, which means that the students spread differently in the floors of the library as they use less frequently the rooms to discuss with fellow classmates for projects. This pessimistic approach generates a lower bound on the accuracy of the WiFi system because $e^{(\kappa)} \neq 0$ and $e^{(\kappa)}$ is an error term linked to the partial coverage that does not usually appear for ideal installations. In other words, any system with limited coverage would be subject to noise $e^{(\kappa)}$ and all systems with full coverage have $e^{(\kappa)} = 0$. \\

Figure~\ref{fig:ResOneExtrapFactorTimeRest} compares Affluences counts against WiFi ones, for the restricted time frame running from 09:00 AM to 6:00 PM. Figure~\ref{fig:ResOneExtrapFactor} does the same but displays the full days, which makes the plot easier to read.\\

The estimated global extrapolation factor is equal to $5.031$ (for time frames of $t_f = 60$ seconds), see~(\ref{eq:globalEqFactorEstimate}). Comparatively, in our previous studies with full coverage \cite{determe2020forecasting, singh2020crowd}, we obtained a value of $3$ for time frames of $t_f = 30$ seconds (which is equivalent to an extrapolation factor of $1.5$ with $t_f = 60$ seconds). This suggests $\kappa \simeq 5/1.5 \simeq 3.33$, which is realistic given the coverage (three floors out of eight).\\ 

\begin{figure}[h]
	\centering
	\includegraphics[scale=0.4]{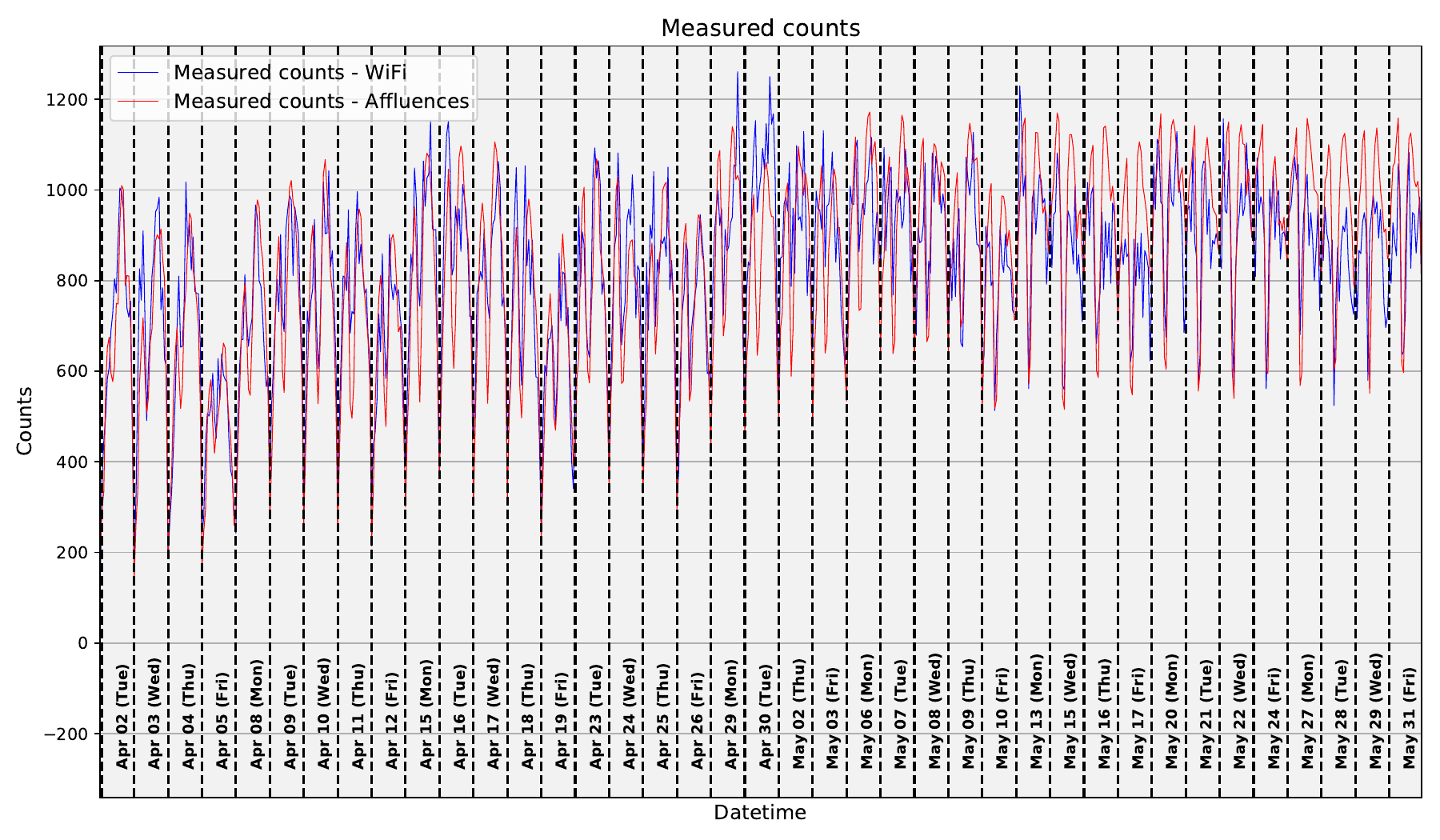}
	\caption{Comparison from 9:00 AM to 6:00 PM of camera and WiFi counts on selected days, with a global extrapolation factor estimate of 5.031, obtained as described in Section~\ref{subsec:estGlobalExtrapFactor}. Affluences refers to a third-party camera counting system and is a ground truth.}
	\label{fig:ResOneExtrapFactorTimeRest}
\end{figure}

\begin{figure}[h]
	\centering
	\includegraphics[scale=0.40]{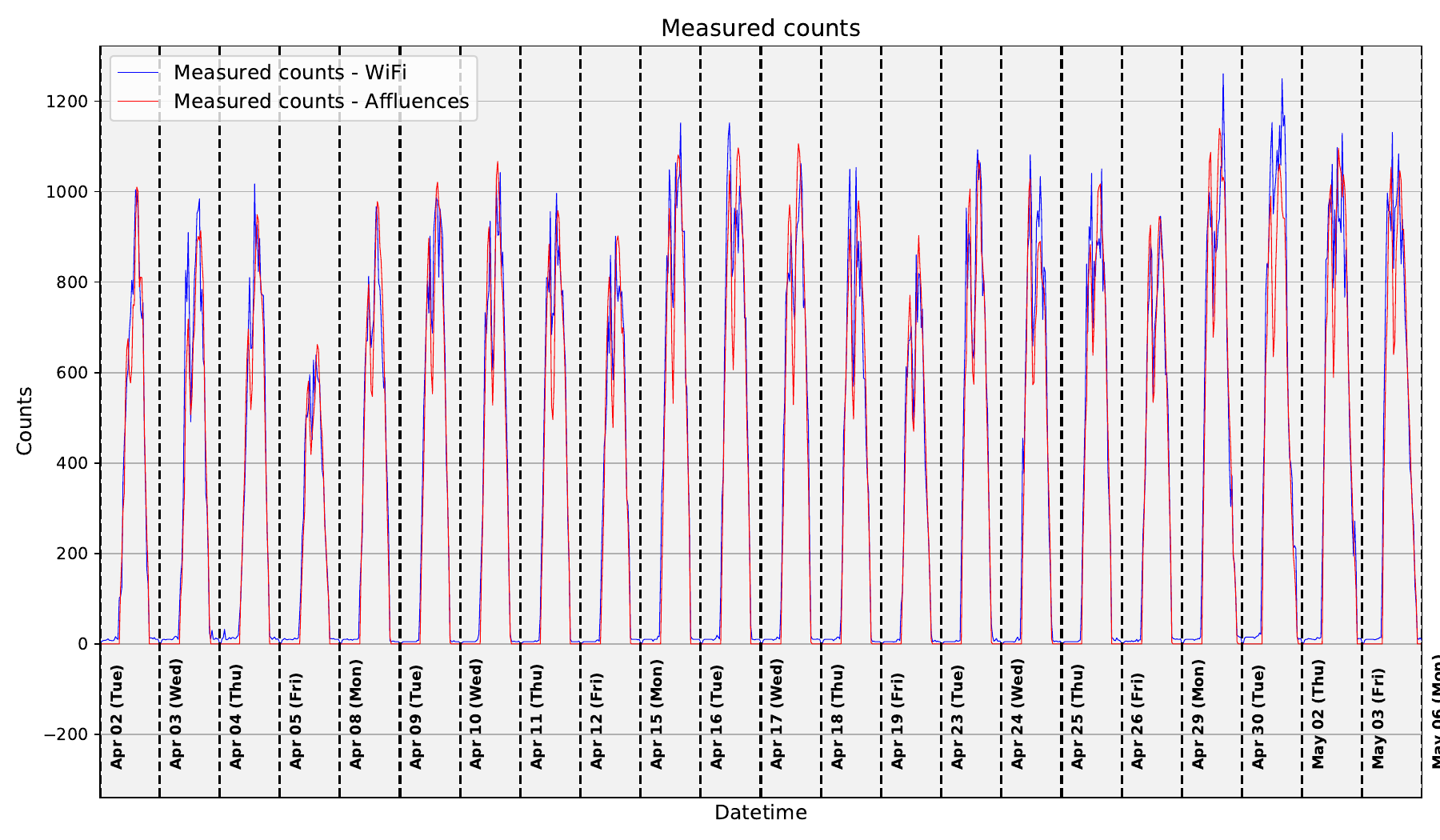}
	\includegraphics[scale=0.40]{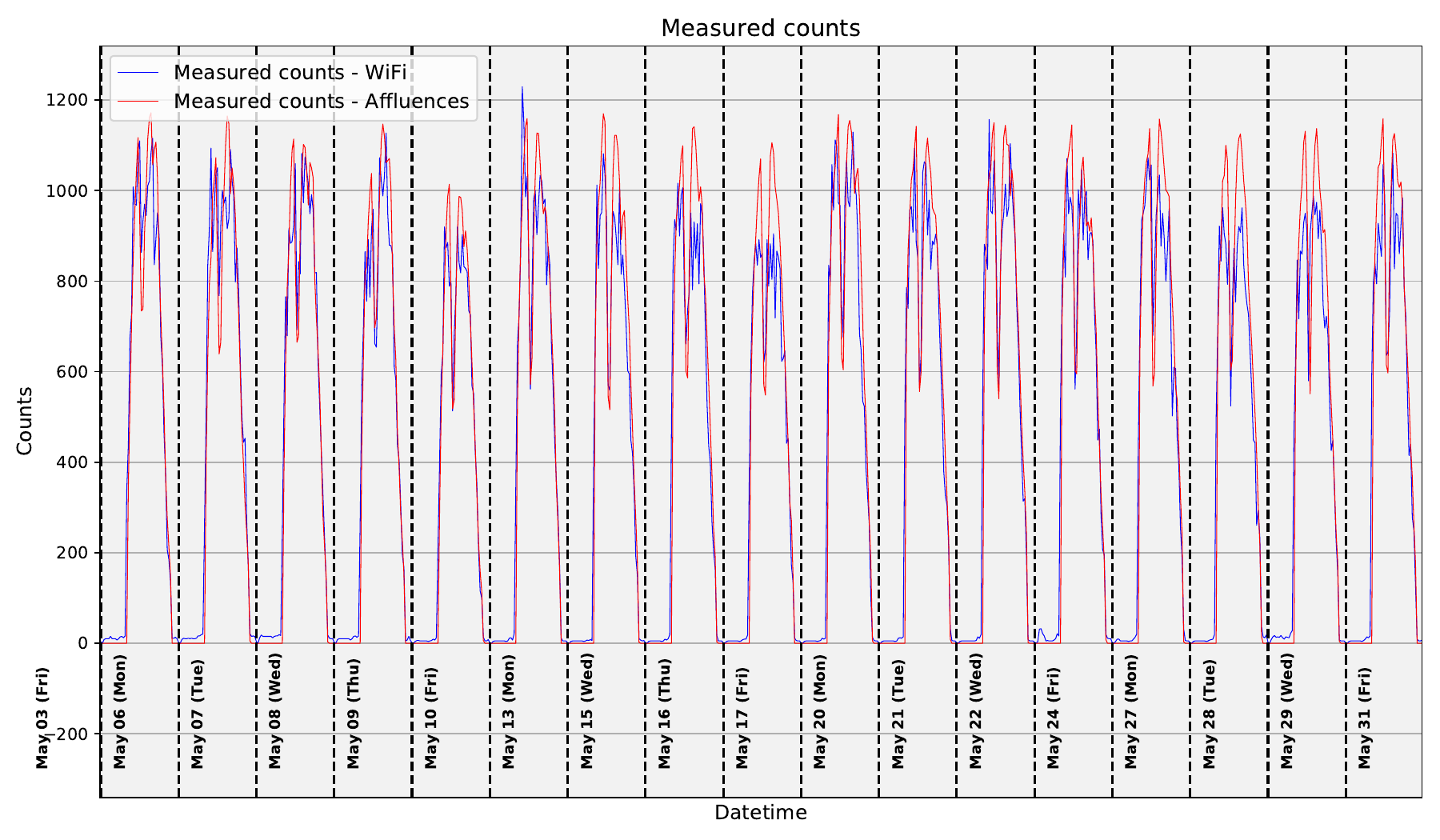}
	\caption{Full comparison of camera and WiFi counts on selected days, with a global extrapolation factor of 5.031. Affluences refers to a third-party camera counting system and is a ground truth.}
	\label{fig:ResOneExtrapFactor}
\end{figure}

For Figure~\ref{fig:ResOneExtrapFactorTimeRest}, the RMSE and MAPE values are of 120.9 and 12.7~\%, respectively. The mean of the counts is equal to 824. These figures are thus upper bounds on the error of the WiFi system.\\

The accuracy estimate based on indoor measurements is pessimistic for large events or buildings because the experiment we could carry out suffers from errors linked to:
\begin{enumerate}
	\item the relatively low number of people in the monitored area (300 people on the three stories against thousands in larger events)
	\item the extrapolation of the crowd counts from three floors to eight floors
	\item the use of RSSI thresholds that we have been tuned coarsely. In large events or using directional WiFi antennas, however, such thresholds would not be necessary as the whole area is large and surrounding areas do not host a significant number of attendees.
\end{enumerate}

\subsection{Results with weekly global extrapolation factors}

This final subsection estimates the global extrapolation factor for each week separately, to better reflect the time-varying distribution of the students across the different floors. Mathematically, it translates into a partial extrapolation factor $\kappa$ in~(\ref{eq:kappaBetaExtrap}) being a function of time. Again, the time-varying nature of the extrapolation factor stems purely from the monitoring of a sub-area and extrapolation of counts to the full area. The global extrapolation factor $\beta$ is constant for events or buildings that are fully covered. \\

Albeit a rather theoretical exercise, compensating the time-varying nature of the extrapolation factor gets accuracy figures closer to those that would have been obtained with full coverage. The improvements resulting from this exercise also suggest that a partial coverage leads to inflated errors in comparison to full-coverage scenarios.\\

Table~\ref{tab:weeklyExtrapFactorResults} reports the results, including the ones of Sec.~\ref{subsec:resultsGlobalExtrapFactor} on its last row. It shows that the global extrapolation factor estimates increase over time, which stems from the humanities library becoming more crowded as examination sessions get closer. A possible explanation is that students favor working on floors that happen to be covered by the sensors and move to the remaining floors as seating options become scarcer; thus, the global extrapolation factor increases over time. \\

Finally, as expected, using extrapolation factors optimized per week improves the average RMSE and MAPE in comparison to using the one obtained for the global, aggregated time series. Nevertheless, the RMSE and MAPE improvements stemming from using weekly extrapolation factors are lower than 10 \%.

\begin{table}
	\centering
	\caption{In ``average", all weeks are weighted identically (that is, without taking into account that some weeks comprise only four days). ``Global time series" corresponds to statistics obtained on the aggregated time series, as described in Sec.~\ref{subsec:resultsGlobalExtrapFactor}.}
	\begin{tabular}{c || c | c | c | c}
		Week starting on & $\textrm{Estimate} \lbrack \tilde{\beta} \rbrack$ & Mean of counts & RMSE & MAPE  \\
		\hline \hline
		2019-04-01 & 4.75 & 604 & 90.3 & 13.4 \% \\
		2019-04-08 & 4.89 & 708 & 94.1 & 12.0 \% \\
		2019-04-15 & 4.84 & 752 & 108.2 & 12.9 \% \\
		2019-04-22 & 4.79 & 752 & 99.5 & 10.5 \% \\
		2019-04-29 & 4.66 & 862 & 138.0 & 13.6 \% \\
		2019-05-06 & 5.11 & 913 & 120.0 & 11.3 \% \\
		2019-05-13 & 5.37 & 934 & 120.0 & 10.9 \% \\
		2019-05-20 & 5.24 & 962 & 105.6 & 9.2 \% \\
		2019-05-27 & 5.47 & 954 & 124.3 & 11.5 \% \\
		\hline
		Average & 5.01 & 827  & 111.1  & 11.7 \% \\
		\hline \hline 
		Global time series & 5.03 & 824 & 120.9 & 12.7 \%
	\end{tabular}
	\label{tab:weeklyExtrapFactorResults}
\end{table}

\section{Practical considerations when deploying sensors} \label{sec:practicalConsiderations}

In public events, our experience is that sensors are often not connected to a dedicated power supply line, sharing instead power supplies with other devices (e.g., lightning devices). These other circuits may be unplugged to save power at night or during daytime. Even if the sensors were connected to dedicated circuits, these could malfunction or be shut down for maintenance without prior notice. Therefore, we recommend making sensors unaffected by improper shutdowns, by using high-quality persistent storage (e.g., using high-end eMMC memory) and by mounting the operating system in read-only mode.

\section{A comparison of the present counting system with its contenders} \label{sec:comparisonContenders}
Before reaching the conclusion, it is important to compare the present system with its main contenders listed in Table~\ref{tab:comparisonWorks} of Section~\ref{subsec:stateoftheart}. In particular, accuracy is an interesting basis of comparison. The accuracy figure obtained in this work (of about 12 \%) is comparable or better than all the listed works except for \cite{li2018experimental}, which is a cooperative system in that it requires individuals to connect their WiFi devices to access points. Moreover, \cite{li2018experimental} has not been tested for areas hosting hundreds of individuals. The work \cite{donelli2021crowd} has sometimes better accuracy and sometimes worse accuracy than the present counting system; it also has not been tested on crowds of more than 100 people (see \cite[Fig.~6]{donelli2021crowd}). Of course, it is always dangerous to compare accuracy figures without testing all counting systems on a common monitoring area. Unfortunately, such an endeavor is impossible to carry out given that many of the contenders are novel solutions that are not yet commercially available and are expensive and time-consuming to reimplement (with inevitable implementation differences anyway). Nevertheless, accuracy comparisons make it possible to determine whether the accuracy of different counting systems are similar, which appears to be the case here. \\

As a conclusion, the counting system that this manuscript presents is a strong contender in comparison to the other existing systems, especially for large crowds (at least a few hundreds of people or more). It does not depend on user cooperation and has been experimentally validated on large crowds. Moreover, it does not require costly equipment and the required sensor density (of about one sensor per 25 m x 25 m = 400 m$^2$ for dense crowds) is comparable or lower than that of other counting systems (in particular, it is an order of magnitude below the sensor density required for \cite{denis2020large, kaya2020large, denis2020sensing}, which ranges from one sensor/(15 m$^2$) to one sensor/(40 m$^2$)). Finally, among non-cooperative systems, our accuracy figures are competitive.\\

This work is also unique in that it provides a statistical model of the counting process and derives a concentration inequality that shows its relative accuracy increases with the number of monitored individuals. In this sense, it offers some degree of theoretical validation.

\section{Conclusion} \label{sec:conclusion}

This paper describes a crowd monitoring system relying on probe requests transmitted by attendees' smartphones in the monitored area. This system is suitable for indoor and outdoor areas hosting at least a few hundreds of attendees. The monitoring system ensures strict privacy requirements are met and is therefore compatible with modern privacy laws. We provided both theoretical and experimental evidence that our system computes accurate estimates of the number of attendees. Despite non-ideal experimental conditions, the MAPE we computed is of less than 13 \%.

\section*{Acknowledgments}
We thank Innoviris for funding this research through the MUFINS project and Brussels Major Events for their active collaboration. We also thank the IT team working at the ULB's Humanities library. Finally, we are grateful to the Icity.Brussels project and FEDER/EFRO grant for their support. 

% \nocite{*}
\bibliographystyle{IEEEtran}
\bibliography{mybib}

\end{document}